\documentclass{article}
\usepackage{psfrag,graphicx,amssymb,amsmath,theorem}
\newtheorem{theorem}{Theorem}
\newtheorem{lemma}[theorem]{Lemma}
\newtheorem{cor}[theorem]{Corollary}

\def\QED{\ensuremath{{\square}}}
\def\markatright#1{\leavevmode\unskip\nobreak\quad\hspace*{\fill}{#1}}
\newenvironment{proof}
  {\begin{trivlist}\item[\hskip\labelsep{\bf Proof.}]}
  {\markatright{\QED}\end{trivlist}}

\newcommand{\Conv}{\mathit{Conv}}
\title{Hamiltonian Tetrahedralizations with Steiner Points}

\author{
Francisco Escalona\footnotemark[1] \and
Ruy Fabila-Monroy\footnotemark[2] \and
Jorge Urrutia  \footnotemark[2]
}

\index{Escalona, Francisco}
\index{Fabila-Monroy, Ruy}
\index{Urrutia, Jorge}

\begin{document}
\maketitle

\begin{abstract}

Let $S$ be a set of $n$ points in $3$-dimensional 
space. A tetrahedralization 
$\mathcal{T}$ of $S$ is a set of interior 
disjoint tetrahedra with vertices on $S$, 
not containing points of $S$ in their interior, 
and such that their union is the convex hull of $S$.
Given $\mathcal{T}$, $D_\mathcal{T}$ is
defined as the graph having as vertex set
the tetrahedra of $\mathcal{T}$, two of which are adjacent
if they share a face. We say that $\mathcal{T}$ 
is Hamiltonian if  $D_\mathcal{T}$ has
a  Hamiltonian path. Let $m$ be the number
of convex hull vertices of $S$.
We prove that by adding
at most $\lfloor \frac{m-2}{2} \rfloor$ 
Steiner points to interior of the convex hull of $S$, 
we can obtain a point
set that admits a Hamiltonian tetrahedralization.
An $O(m^\frac{3}{2}) +  O(n \log n)$  time
algorithm to obtain these points is given.
We also show that all point sets with at most $20$
convex hull points admit a Hamiltonian tetrahedralization
without the addition of any Steiner points.
Finally we exhibit a set of $84$ points that does not
admit a Hamiltonian tetrahedralization in which all
tetrahedra share a vertex.
\end{abstract}

\renewcommand{\thefootnote}{\fnsymbol{footnote}}

\footnotetext[1]{Facultad de Ciencias, Universidad Nacional Aut\'onoma de M\'exico}
\footnotetext[2]{Instituto de Matem\'aticas, Universidad Nacional Aut\'onoma 
de M\'exico (\texttt{ruy@ciencias.unam.mx}, 
\texttt{urrutia@math.unam.mx}). Supported by
CONACYT of Mexico, Proyecto SEP-2004-Co1-45876, and 
PAPIIT (UNAM), Proyecto IN110802.}

\renewcommand{\thefootnote}{\arabic{footnote}}

\section{Introduction}

All point sets considered throughout this paper will be in general position
in $\mathbb{R}^2$ and $\mathbb{R}^3$.
This are point sets such that: in $\mathbb{R}^2$ not three
of its elements are colinear and in $\mathbb{R}^3$ not
four of its elements are coplanar.\par

Let $S$ be a set of $n$ points in $\mathbb{R}^3$. 
A tetrahedralization $\mathcal{T}$ of $S$ is a set 
of tetrahedra with vertices in $S$, such that :

\begin{enumerate}
\item Their union is the convex hull of $S$.

 \item The tetrahedra only intersect at points, lines
or faces.

    \item The tetrahedra do not contain points of $S$ in
their interior.
\end{enumerate}

Given $\mathcal{T}$, we define the dual
graph of $\mathcal{T}$, $D_\mathcal{T}$ to
be the graph whose vertices are the elements of 
$\mathcal{T}$, two of which are adjacent if they
share a common face.\par

In a similar way a triangulation of a point set of points  $S$ 
in the plane is set of interior 
disjoint triangles
with vertices on $S$, not containing points of $S$
and such that their union is the convex hull of $S$.\par

Again, the dual graph $D_\mathcal{T}$ of $\mathcal{T}$ is the graph having
 the elements of $\mathcal{T}$ as vertices, two of which
are adjacent if they share and edge.\par

A Hamiltonian path (cycle) of a graph $G$ is a path (cycle)
spanning all the vertices of $G$. If $D_\mathcal{T}$ contains
a Hamiltonian path or cycle, we say that $\mathcal{T}$
is a Hamiltonian tetrahedralization (or triangulation, if $S$ is in
the plane).\par

The problem of finding a Hamiltonian triangulation for
a given point set in the plane, has been
settled in both the existential and algorithmic
sense: every set of $n$ points in the plane 
admits a hamiltonian 
triangulation and it can be computed
in time $O(n\log n)$ \cite{skiena, faurr}.

The Hamiltonian Tetrahedralization Problem 
\cite{skiena}, Problem 29 in \cite{open}, 
has been a long standing problem in 
Computational Geometry.
It is not known if every point set in general
position in $\mathbb{R}^3$ admits
a Hamiltonian tetrahedralization.
It is conjectured in \cite{skiena}, that the
problem of finding such a tetrahedralization
is $NP$-hard for arbitrary point sets.\par

Hamiltonian triangulations were initially studied,
among other reasons, because they speed
the rate at which a triangulation can be sent
to a graphic processor for rendering (\cite{skiena, spacetradeoffs}).
A similar speed up applies for tetrahedralizations.\par

In this paper we study the problem of computing
Hamiltonian tetrahedalizations by adding Steiner 
points.\par

We call the points in the interior of the convex hull of 
$S$, \emph{interior points} and the points on the
boundary of the convex hull, \emph{exterior points}
Let $m$ and $m'$ be the number of exterior and
interior points of $S$ respectively; 
we denote the convex hull of $S$ by $Conv(S)$.\par

Our main result is thus:

\begin{theorem}\label{teo:main}
Let $S$ be a set of $n$ points in general position with $m$ exterior points in $\mathbb{R}^3$. Then we need to add at most $\lfloor \frac{m-2}{2} \rfloor$ Steiner points to the interior
of  $\Conv(S)$, so that the resulting point
set admits a Hamiltonian tetrahedralization. Moreover this tetrahedralization can
be found in time $O(m^{\frac{3}{2}})+O(n\log n)$
\end{theorem}

We note that the Steiner points are added as interior points of $S$, it is not
hard to see that if this requirement is dropped a Hamiltonian tetrahedralization
can be found by adding only two Steiner points.\par

The paper is organized as follows:\par

In Section \ref{join} we present an algorithm that adds at most $\lfloor \frac{m-2}{2} \rfloor$
Steiner points,  to the interior of $\Conv(S)$.
Our algorithm produces a Hamiltonian
tetrahedralization.
The overall complexity
of the algorithm is $O(m^\frac{3}{2}) +  O(n \log n)$. We consider its complexity
and implementation issues in Section \ref{alg}.
In Section \ref{section:3ccp} we study the dual graph of the convex hull
of $S$. We show that all point sets with at most $20$ exterior points 
admit a Hamiltonian path tetrahedralization. In the same section a lower bound 
on the number of Steiner points our algorithm might add is given. With the same techniques
we improve on the result of  \cite{isora} and exhibit a point set of $84$ elements
that does not admit a Hamiltonian pulling tetrahedralization (a pulling tetrahedralization
is a tetrahedralization in which all tetrahedra share a point).
Finally in Section \ref{section:conclusions} a summary of results is given and new directions offered.
\par

\section{The algorithm}\label{join}

In this section we sketch an algorithm that achieves Theorem \ref{teo:main}.\par

The main idea of our algorithm is to first add a point to $S$  to obtain
a tetrahedralization such that its dual graph can be
partitioned into cycles.\par

Steiner points are then inserted to join existing 
cycles. We continue this process until
the cycle partition consists of just one cycle. This final
cycle is a Hamiltonian cycle in the dual graph of the final
tetrahedralization.\par

The first step is to remove the interior points and 
those exterior points of degree $3$ (that is, points
adjacent to $3$ other points in the boundary of $\Conv(S)$).
We can do this in view of the following: 

\begin{lemma}\label{obsconv}
  If the set of exterior points of $S$ admit a Hamiltonian 
  tetrahedralization, so does $S$.
 \end{lemma}
 
\begin{proof}
Consider an interior point $x$ of $S$ and suppose 
$S-\{x\}$ admits a Hamiltonian tetrahedralization $\mathcal{T}$.
Let $\tau$ be the unique tetrahedron of $\mathcal{T}$ that 
contains $x$ in its interior. If we remove $\tau$ from
$\mathcal{T}$ and add the four tetrahedra induced by
the faces of $\tau$ with $x$, 
we obtain a tetrahedralization of $S$ and 
the Hamiltonian cycle of  $D_{\mathcal{T}}$ can be extended 
to a Hamiltonian cycle of the new tetrahedralization.
Applying this process recursively the result follows.
\end{proof} 

Assume thus that $S$ does not have interior points.

\begin{lemma} \label{theo1}
Let $x$ be an exterior point of $S$ of degree $3$. If 
$S-\{x\}$ admits a Hamiltonian tetrahedralization, then so does
$S$.
\end{lemma}

\begin{proof}
Suppose $S-\{x\}$ admits a Hamiltonian tetrahedralization 
$\mathcal{T}$. The three convex hull vertices of $S$ 
adjacent to $x$ form a face  $F$ of the boundary of
 $\Conv(S-\{x\})$.
Let $\tau_1$ be the only tetrahedron of $\mathcal{T}$ that 
contains $F$ as a face and let $\tau_2$ be the 
tetrahedron induced by $x$ and $F$. Clearly $\tau_1 \cup \tau_2$
is convex.
If we remove $\tau_1$ and  $\tau_2$ from $\mathcal{T}$
 and replace them
with the three tetrahedra induced by the faces of $\tau_1$
(except $F$) and $x$, we obtain a tetrahedralization 
$\mathcal{T}'$ of $S$. The Hamiltonian cycle of 
$D_\mathcal{T}$ can now be extended to a Hamiltonian cycle
of $D_{\mathcal{T}'}$.
\end{proof}

Assume now  that $S$ does not contains 
exterior points of degree $3$.\par

We insert a point $p_0$ in the interior of $\Conv(S)$ and join
every face of the boundary of $\Conv(S)$ to $p_0$, forming 
a tetrahedralization $\mathcal{T}$ of $S \cup \{ p_0\} $.\par

Let $G$ be the graph induced by the $1$-skeleton of 
the boundary of $\Conv(S)$, that is, the graph whose vertex set
consists of the exterior points of $S$ and whose edges are the edges
of the boundary of $\Conv(S)$. It is easy to see that both $G$ and its dual graph are planar 
and $3$-connected.
By construction, the dual graph of $G$ is isomorphic to 
$D_\mathcal{T}$. Since every face of $G$ is a triangle, 
$D_\mathcal{T}$ is a regular graph of degree $3$.\par

To obtain the initial partition, we use a theorem of
Petersen \cite{petersen} that states that every $2$-connected
cubic graph contains a perfect matching. Since
$D_\mathcal{T}$ is $3$-connected, in particular it is
 $2$-connected and therefore contains a perfect matching $M$. 
If we remove the edges of $M$ from $D_\mathcal{T}$, we obtain
a regular graph of degree $2$. This subgraph of $D_\mathcal{T}$
is the initial cycle partition.\par

\begin{figure} \label{steinerfig}
  \begin{center}
    \includegraphics[width=0.80\textwidth]{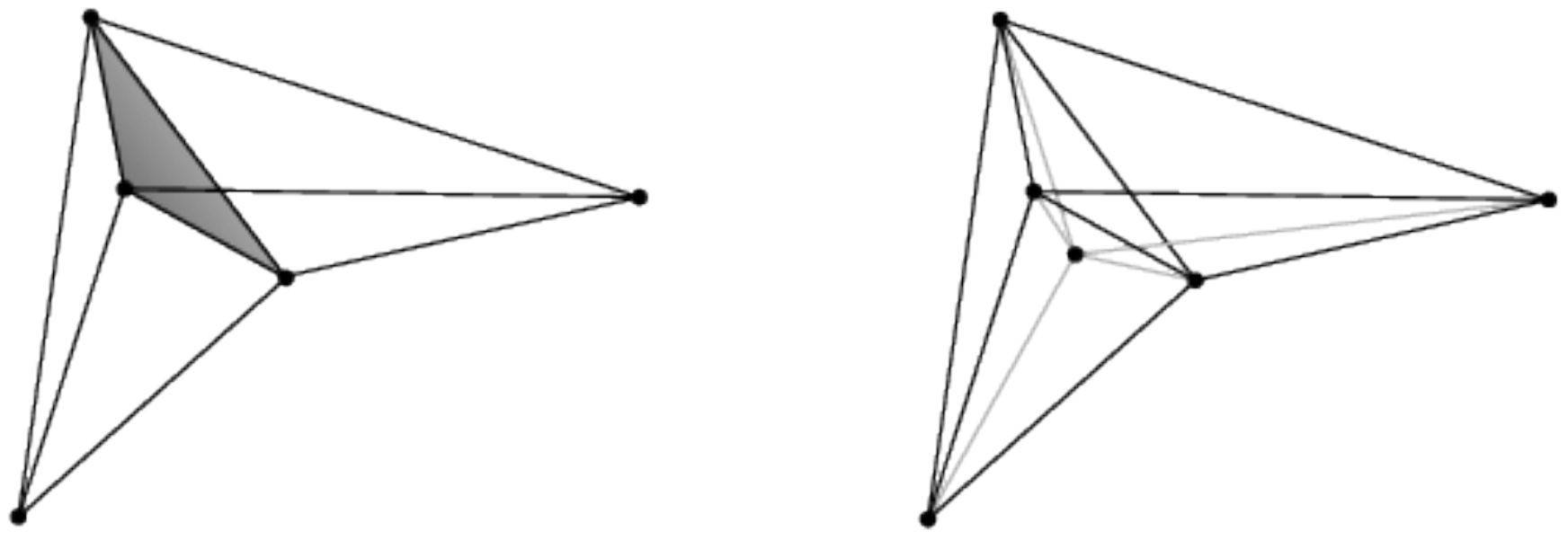}
  \end{center}
    \caption{Join Operation.}
\end{figure}

\subsection{Joining cycles} \label{joincycles}

Consider two disjoint cycles,
$C_1$ and $C_2$, in our cycle partition of $D_\mathcal{T}$, and 
 suppose that there is an edge $e$ of $D_\mathcal{T}$
 that has its end points $\tau_1$ and 
$\tau_2$  in $C_1$ and $C_2$ 
respectively. Since $\tau_1$ and $\tau_2$ are tetrahedra in
$\mathcal{T}$, $e$ corresponds to a shared face $F$ of 
$\tau_1$ and $\tau_2$.\par

The join operation consists of adding a point $p$ to 
the interior of $\tau_1$ so that the line segment
joining the point $q$ in $\tau_2$ opposite to $F$ in $\tau_2$
intersects $F$. We now remove  $\tau_1$ and
$\tau_2$ and replace them by the six tetrahedra 
induced by the faces of  $\tau_1$,
$\tau_2$ and $p$ (except $F$) as shown in  Figure \ref{steinerfig}.

$C_1$ and $C_2$ are joined into a cycle passing the tetrahedra of 
$C_1 \cup C_2 - \{\tau_1, \tau_2\}$ plus
the six new tetrahedra containing $p$ as a vertex
(see Figure 2).\par

\begin{figure} \label{dualsteiner}
  \begin{center}
	\psfrag{1}[][][1]{$C_1$}
    \psfrag{2}[][][1]{$C_2$}
    \psfrag{c}[][][1]{$C$}					
	\includegraphics[width=0.5\textwidth]{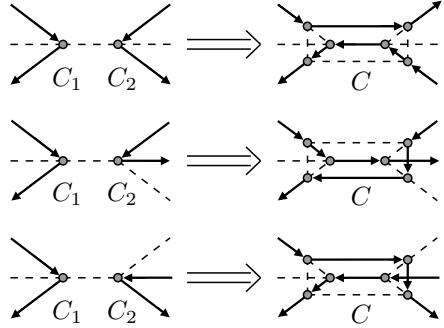}
  \end{center}
    \caption{$D_\mathcal{T}$ before and after the join operation.}
\end{figure}

We repeat this process until a single cycle is obtained.
We will show in the next section that the number of 
Steiner points we need to insert before  a Hamiltonian
cycle is reached is at most $\lfloor \frac{m-2}{2} \rfloor$.

\section{Complexity and implementation.} \label{alg}

In this section we will analyze the running time and 
implementation issues of the algorithm sketched in Section 
\ref{join}.\par

Suppose now that $S$ is a point set with $n$ points in
$\mathbb{R}^3$ with $m$ exterior points and $m'$ interior 
points, $m+m'=n$.
We first calculate the convex hull of 
$S$ in $O(n \log n)$, and then remove the interior points of $S$.\par

Next, we remove the exterior vertices of degree $3$.
This can be done in $O(m)$ by using a priority queue with
all exterior points of degree $3$. Each time one is removed,
the degree of its neighbors is checked and if necessary they
are added to the queue.\par

Adding the first Steiner point $p_0$ and tetrahedralizing 
as in the previous section takes time $O(m)$.\par

The complexity of finding the initial cycle partition described
at the end of Section~\ref{join}
is that of finding 
a perfect matching in $G$. In a graph with $|V|$ vertices and $|E|$ 
edges, a perfect matching can be found  in time 
$O(|E| \sqrt{ |V|})$~\cite{micali}. Since we are dealing
with a cubic graph, we have $|E|=\frac{3}{2} |V|$. Thus we can 
find the initial cycle cover in 
$O(\frac{3}{2} m \sqrt{m})=O(m^\frac{3}{2})$ time.\par

Once we have the initial cycle cover, we return the vertices 
that were removed. This is done before the 
join operations in order to take advantage of the structure
of the tetrahedralization and return the exterior points of
degree $3$ and interior points efficiently.
Using the fact that $D_\mathcal{T}$ is a planar graph, the 
interior points and exterior points of degree $3$ can be 
added using point location at a cost of $O(\log m)$ per point.
The exterior points
of degree $3$ are added first and the interior points afterwards.
As these points are returned, the initial cycle partition is
updated as in Lemma~\ref{obsconv} and Theorem~\ref{theo1}.\par

We have to be careful about the order 
in which the interior points are added.
Suppose we have a tetrahedra $\tau$ which contains $k$ 
interior points that remain to be added, and that 
we return one of these points.  When we 
retetrahedralize the point set, $\tau$ would be 
split  into $4$ new tetrahedra.
We have to guarantee that each of these tetrahedra 
receives a linear fraction of the points
in $\tau$, for otherwise the 
iterative process could take as much as 
$O(k^2)$. That is, we need a splitter vertex (see~\cite{avis}).
Such a vertex can be found in time $O(k)$, thus ensuring a 
total of $O(m'\log m)$ running time.\par

Finally we proceed to merge the set of cycles obtained so
far into a single cycle as in Subsection~\ref{joincycles}.
Each time we join two cycles, we insert one Steiner point.
Since $G$ has $m$ vertices, the number of faces of $G$
is $2m-4$, and since all the cycles obtained have at least
four vertices, the initial cycle partition contains at most 
$\lfloor \frac{2m-4}{4} \rfloor$ elements. Thus the number of
Steiner points required is at most $\lfloor \frac{m-2}{2} \rfloor$.
Since there are $O(m)$ edges 
in $G$ the merging of the cycles can be done in time $O(m)$.
The overall complexity of the algorithm is thus 
$O(m^\frac{3}{2}) +  O(n \log n)$. From this and all previous observations
we obtain Theorem \ref{teo:main}.\par


\section{$3$-connected cubic planar graphs}\label{section:3ccp}

To conclude the paper we study tetrahedralizations of 
point sets in terms of the dual graph of
the 1-skeleton of their convex hull. The $1$-skeleton of 
the convex hull of $S$ is the graph having as vertices
the exterior points of $S$, two of which are adjacent
if they are joined by an edge in the boundary of $Conv(S)$.
 In particular
we prove that every point set of at most $20$ exterior
points admits a Hamiltonian tetrahedralization. Also
a set of $84$ points that does not contain a Hamiltonian pulling tetrahedralization
is shown. Improving therefore on the result of \cite{isora}. 
The techniques employed, allow us to give a lower bound on the number
of Steiner points our algorithm might add.\par

The convex hull of a point set in $\mathbb{R}^3$ 
is a convex polyhedron of triangular faces. It is known that 
the dual graph of such a polyhedron is a $3$-connected cubic planar graph ($3$CCP).
The converse is also true, that is every  $3$CCP graph can be
realized as the dual graph of some polyhedron and therefore as the 
dual graph of the $1$-skeleton of the convex hull of a point set in convex
position  \cite{steinitz}.\par

$3$CCP graphs are uniquely embeddable in the plane \cite{whitney}. This
in particular means that faces of a $3$CPP graph are defined regardless of any particular 
embedding. Now, given a polyhedron $\mathcal{P}$, take any point $q$ in the interior 
of a face and do a stereo-graphic projection to a plane not containing $\mathcal{P}$
and such that any line segment joining a point of the plane and $q$ cuts
$\mathcal{P}$ in its interior. This yields a planar embedding of the dual graph 
$D(\mathcal{P})$ of $\mathcal{P}$. $D(\mathcal{P})$ is the graph having as vertices
the faces of $\mathcal{P}$ two of the adjacent if they share an edge.
Note that in this embedding: faces of
$D(\mathcal{P})$ correspond to all the faces of $\mathcal{P}$ containing a given point 
and all the faces of $\mathcal{P}$ containing a point correspond to faces of 
$D(\mathcal{P})$. Since all embeddings are essentially unique we may assume that this is 
always the case for any embedding of $D(\mathcal{P})$.\par

$3$CCP graphs were once conjectured to be Hamiltonian by Tait \cite{tait}, until a
$3$CCP non-Hamiltonian graph of 
$46$ vertices was found by Tutte \cite{tutte}. Nevertheless
using an exhaustive computer search it has been shown that the smallest non-Hamiltonian $3$CCP graphs have $38$ vertices \cite{nonham}, in other words all $3$CCP graphs of at most 36 vertices
have a Hamiltonian cycle.
We use this fact to show that all point sets of at most $20$ elements admit
a Hamiltonian tetrahedralization.\par

At this point it should be stressed that we are making no assumption
on the existence of interior points. Indeed if the dual graph of the
convex hull of a point set contains a Hamiltonian cycle then joining all the vertices
of the convex hull to an interior point would yield a tetrahedralization whose dual graph
is isomorphic to the dual graph of the  and thus would contain a Hamiltonian cycle.

\begin{theorem}\label{eyeham}
Let $G$ be a $3$CCP Hamiltonian graph. There exists a face $F$ 
of $G$ so that $G-F$ contains a Hamiltonian path
\end{theorem}
\begin{proof}

Consider a planar embedding of $G$ and a Hamiltonian
cycle $C$ of $G$. Define the distance of two vertices 
as the minimum length of the two paths joining them in $C$.
Take any edge $e$ in $G$ not in $C$, joining vertices two vertices $x$ and $y$,
 whose distance in $C$ is minimum. The path $\Gamma$ in $C$ joining $x$ and $y$, realizing this
distance together with $e$ forms a face $F$; since any other edge
would join vertices at a strictly less distance.
Now $C$-$\Gamma$ is a Hamiltonian path of $G-F$.
\end{proof}

Take any vertex $p$ of $\mathcal{P}$ and consider the tetrahedralization
$T_p$ formed by joining all other vertices of $\mathcal{P}$ to $p$.
Such tetrahedralizations are known in the literature as ``pulling'' tetrahedralizations.
Let $F_p$ be the corresponding face of $p$ in $D(\mathcal{P})$. $T_p$ is 
isomorphic to $D(\mathcal{P})-F_p$. Note that if $D(\mathcal{P})$ is Hamiltonian, 
Theorem \ref{eyeham} implies the existence of a point $p$ such that $T_p$ is Hamiltonian.\par

By Euler's formula, a $3$CCP graph on $n$ vertices has $\frac{n+4}{2}$ faces.
Since all $3$CCP graphs of $36$ or less vertices are Hamiltonian we have:

\begin{cor}\label{corol}
Every point set in $\mathbb{R}^3$ in general position of at most
$20$ points admits a Hamiltonian path (``pulling'') tetrahedralization.
\end{cor}

Although it is a simple observation, Theorem \ref{eyeham} serves as a bridge
between $3$CCP graphs and point sets in space. For example it is known that 
$3$CCP graphs of at most $176$ vertices and face size at most $6$ are Hamiltonian.
For point sets this implies that all sets of at most $90$ vertices and with vertices
of degree at most $6$ in its convex hull admit a Hamiltonian tetrahedralizations.
See \cite{mckay} for various similar results on $3$CCP graphs. 
Also a well known conjecture on $3$CCP graphs states that all bipartite $3$CCP
graphs are Hamiltonian \cite{barnette}.\par

Recently, point sets of $92$ with no Hamiltonian path pulling tetrahedralizations have 
been shown to exist \cite{isora}. We improve on this previous result
and exhibit a set with less points ($84$) without a pulling Hamiltonian
path tetrahedralizations.
Our construction also enable us to find lower bounds on the number of Steiner 
points added by the algorithm presented in section \ref{join} and \ref{alg}.\par

\subsection{Blowing up vertices}

In this section we introduce an operation that will allow us to replace any vertex in a $3$CCP graph
with another $3$CCP graph, so that the resulting graph is again a $3$CCP graph.
Using this operation we will construct $3$CPP graphs with certain desired properties.\par

Let $G$ and $H$ be $3$CCP graphs and $v$ be any vertex of $H$.
We may assume that $H$ is embedded in the plane so that $v$ is a vertex on the exterior face.
Let $N_H(v)=\{v_1,v_2, v_3\}$ be the neighbourhood of $v$. Remove $v$ from 
$H$ and add a path of $4$ new vertices $(v_1',v_2',w,v_3')$. Join $v_i$ with $v_i'$; 
call the resulting graph $H'$, see Figure 3.\par

\begin{figure} \label{con3ccp}
  \begin{center}
    \psfrag{H}[][][1]{$H$}
    \psfrag{H1}[][][1]{$H'$}
    \psfrag{v1}[][][1]{$v_1$}
    \psfrag{v2}[][][1]{$v_2$}
    \psfrag{v3}[][][1]{$v_3$}
    \psfrag{v1p}[][][1]{$v_1'$}
    \psfrag{v2p}[][][1]{$v_2'$}
    \psfrag{v3p}[][][1]{$v_3'$}
    \psfrag{w}[][][1]{$w$}
    \psfrag{v}[][][1]{$v$}
    \includegraphics[width=0.8\textwidth]{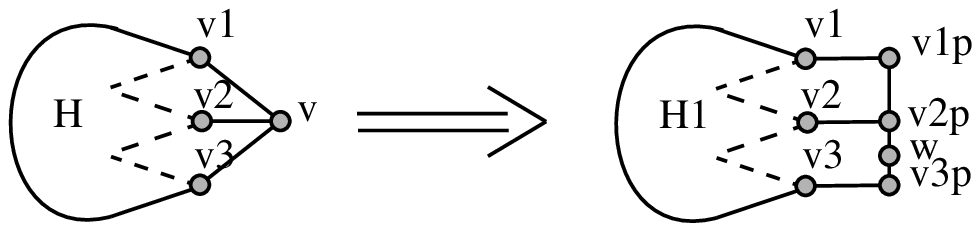}
  \end{center}
    \caption{}
\end{figure}

Although $H'$ is no longer cubic, it can however be used to replace any vertex $u$ of $G$:
Let $N_G(u)=\{u_1, u_2, u_3\}$ with its elements in a given order; Remove $u$ from $G$ and replace it with $H'$;
join $v_i'$ with  $u_i'$
The resulting graph $G'$ is a $3$CCP graph. By choosing an adequate order on the elements of $N_G(u)$, we can place $H$ in a particular face of $G$, see Figure 4. \par

\begin{figure}\label{fig:choice}
	\begin{center}
  		\includegraphics[width=0.8 \textwidth]{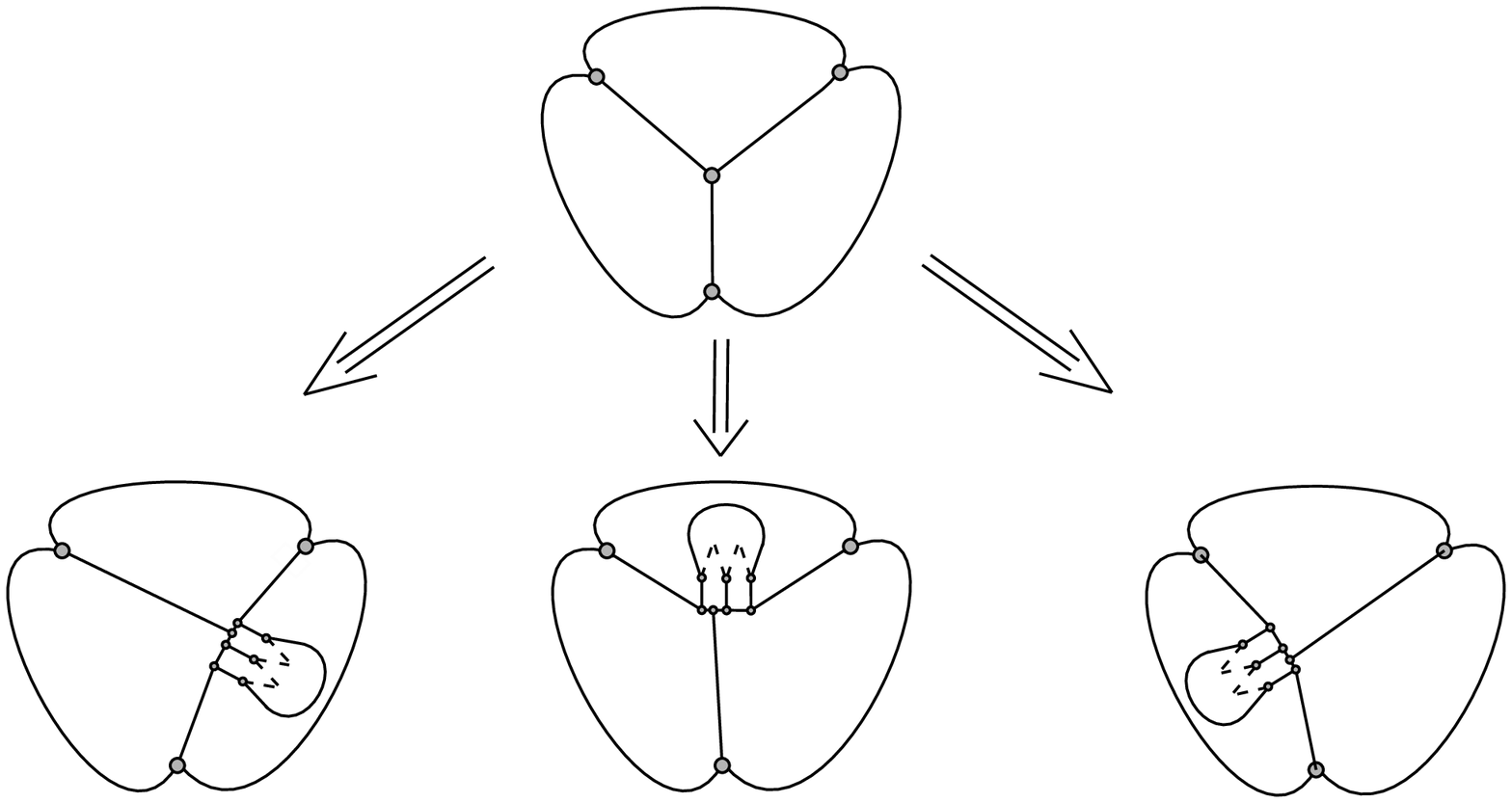}
	\end{center}
	\caption{}
\end{figure}

We use this property to place a non Hamiltonian $3$CCP graph $H$ on each face of $K_4$ (which is
a $3$CCP graph). Call the resulting graph $G$.
The faces of $K_4$ isolate each non Hamiltonian graph. 
Assume that after a removal of a face of $G$ there exists a Hamiltonian path on $G$. 
Since there are $4$ copies of $H$ on $G$
two of them may contain an endpoint of the path, one more may contain the face that was removed.
In the remaining copy of $H$ the Hamiltonian path enters and leaves $H$. This a contradiction 
since from this path we could derive a Hamiltonian cycle in $H$, see Figure 5.\par

\begin{figure}\label{fig:k4}
	\begin{center}
  		\includegraphics[width=0.5\textwidth]{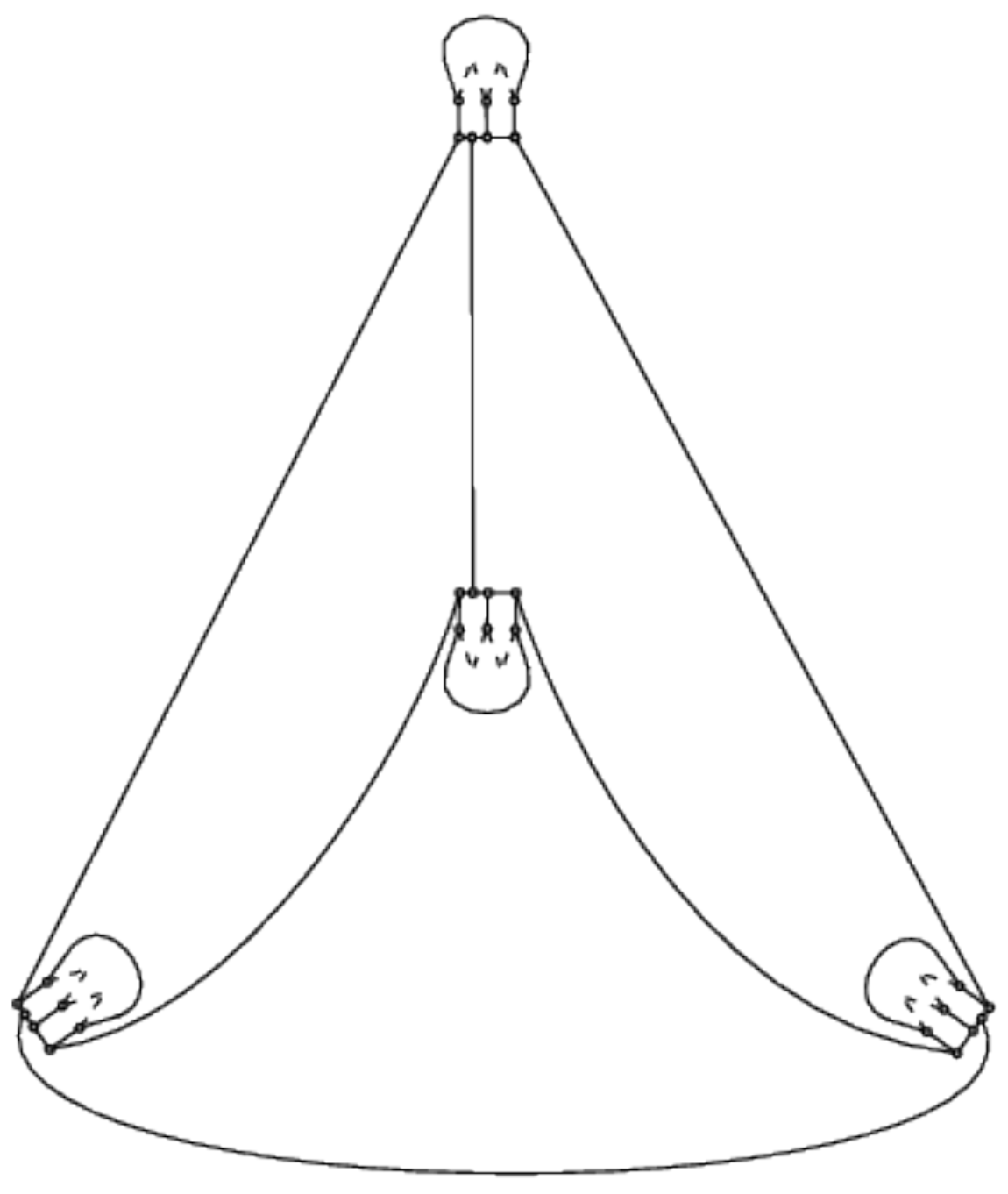}
	\end{center}
	\caption{}
\end{figure}

For $H$ in this construction we may use the smallest non Hamiltonian $3$CCP graph on $38$ vertices.
After each insertion of $H$ into $K_4$, we remove a vertex  and add $41$ new vertices.
 In total $G$ has $41 \times 4=164$
vertices. Therefore there exists a polyhedron of $\frac{164+4}{2}=84$ points in $\mathbb{R}^3$, whose dual graph is isomorphic to $G$. By the observations at the beginning of this section, there is no pulling Hamiltonian tetrahedralization of this polyhedron.\par

Using the same technique we can construct a $3$CCP graph that needs a linear number of disjoint cycles to cover its vertices. Take any $3$CCP graph and replace every vertex with a copy of the smallest
non Hamiltonian graph $H$ of $38$ vertices. In total per vertex $41$ new vertices are added. It is easy to see that any partition of the vertices into cycles in this new graph would need at least a cycle per
 copy of $H$. Therefore since our algorithm adds one Steiner point per cycle, there  are point sets
for which our algorithm adds at least $\frac{n}{41}$ Steiner points.\par

 This gives a lower bound on the number of Steiner points our algorithm might need to find a Hamiltonian tetrahedralization of a given point set.

\section{Conclusions}\label{section:conclusions}

In this paper we considered the problem of computing Hamiltonian Tetrahedralizations
of point sets in $3$-space by adding Steiner points. An algorithm was detailed to do so
for points sets with $n$ points and $m$ exterior points in time  $O(m^\frac{3}{2}) +  O(n \log n)$ .
Our algorithm adds at most $\lfloor \frac{m-2}{2}\rfloor$  Steiner points.\par

It seems natural that there must be a compromise between number of Steiner points added and
the running time of the algorithm employed to do so. A natural question would be to ask:
What is the least number of points that need to be added while maintaining an
efficient algorithm? We conjecture that a sublinear number
of Steiner points cannot be added to obtain a Hamiltonian
tetrahedralization in polynomial time.\par

Regarding the implementation of the algorithm, the algorithms we used
as a subroutines are not straight forward to program. Simpler algorithms
would also be desirable.\par

On the combinatorial side of the problem, it would be interesting to prove
that a sublinear number of Steiner points suffice to obtain a Hamiltonian
Tetrahedralization. Note that the current conjecture is that actually
none are needed. This weakening of the conjecture, nevertheless is worth
studying.\par

Finally we also showed that point sets with at most $20$ points always admit
a Hamiltonian tetrahedralization. Actually this tetrahedralization is a pulling 
tetrahedralization. We also improved the previous upper bound of $92$ to
$84$ on point sets without a Hamiltonian pulling tetrahedralization.
This gap remains to be closed.

\bibliographystyle{plain}
\bibliography{hambib}
\end{document}